\documentclass[sigconf]{acmart}

\usepackage[linesnumbered, titlenumbered, ruled]{algorithm2e}
\usepackage{amsthm}
\usepackage{mathtools}
\usepackage{paralist}
\usepackage{multirow}

\newtheorem{theorem}{Theorem}[section]
\newtheorem{lemma}[theorem]{Lemma}

\AtBeginDocument{%
  }

\setcopyright{acmcopyright}
\copyrightyear{2022}
\acmYear{2022}
\acmDOI{}

\acmConference[KDD-UC]{KDD Undergraduate Consortium}{August 14-18, 2022}{Washington DC, USA}
\acmPrice{15.00}
\acmISBN{978-1-4503-9999-9/18/06}

\begin{document}


\title{Constructing and Sampling Directed Graphs with Linearly Rescaled Degree Matrices}

\author{Yunxiang Yan
}
\authornote{Work performed while at the University of Notre Dame.}
\email{ryan.yunxiang.yan@gatech.edu}
\affiliation{%
  \institution{Georgia Institute of Technology}
  \city{Atlanta}
  \state{Georgia}
  \country{USA}
  \postcode{30332}
}

\author{Meng Jiang}
\email{mjiang2@nd.edu}
\affiliation{%
  \institution{University of Notre Dame}
  \city{Notre Dame}
  \state{Indiana}
  \country{USA}
  \postcode{46556}
}

\begin{abstract}

In recent years, many large directed networks such as online social networks are collected with the help of powerful data engineering and data storage techniques. Analyses of such networks attract significant attention from both the academics and industries. However, analyses of large directed networks are often time-consuming and expensive because the complexities of a lot of graph algorithms are often polynomial with the size of the graph. Hence, sampling algorithms that can generate graphs preserving properties of original graph are of great importance because they can speed up the analysis process. We propose a promising framework to sample directed graphs: Construct a sample graph with linearly rescaled Joint Degree Matrix (JDM) and Degree Correlation Matrix (DCM). Previous work shows that graphs with the same JDM and DCM will have a range of very similar graph properties. We also conduct experiments on real-world datasets to show that the numbers of non-zero entries in JDM and DCM are quite small compared to the number of edges and nodes. Adopting this framework, we propose a novel graph sampling algorithm that can provably preserves in-degree and out-degree distributions, which are two most fundamental properties of a graph. We also prove the upper bound for deviations in the joint degree distribution and degree correlation distribution, which correspond to JDM and DCM. Besides, we prove that the deviations in these distributions are negatively correlated with the sparsity of the JDM and DCM. Considering that these two matrices are always quite sparse, we believe that proposed algorithm will have a better-than-theory performance on real-world large directed networks.

\end{abstract}




\maketitle

\section{Introduction}

Analysis of large directed networks has always been an important topic in data science applications. Such analyses range from the spread of COVID-19 \cite{jo2021social}, twitter fake account detection \cite{ercsahin2017twitter}, to automobile insurance fraud detection \cite{vsubelj2011expert}. However, as researchers show more and more interest in the emerging large directed networks, their attempts to conduct analysis on the whole graph are always hindered by unbearably long running time and running cost because the time complexities of a lot of graph algorithms are polynomial to the size of the graph. One solution is to create a representative sample graph out of the large graph while preserving important properties of the original graph such as degree distributions, degree correlations, and clustering coefficients \cite{lit2020understanding}.

A number of graph sampling algorithms have been proposed in the past. They can be classified into several categories~\cite{leskovec2006sampling,rozemberczki2020little}, such as node sampling~\cite{adamic2001search,stumpf2005subnets,leskovec2006sampling}, edge sampling~\cite{krishnamurthy2005reducing,ahmed2013network}, and exploration-based sampling ~\cite{leskovec2006sampling,hubler2008metropolis,gjoka2010walking,goodman1961snowball,wilson1996generating,leskovec2005graphs,ribeiro2010estimating,maiya2010sampling,lee2012beyond,doerr2013metric,li2015random,zhou2015leveraging,rezvanian2015sampling,rozemberczki2018fast,li2019walking,ricaud2020spikyball}. However, most of them depend on heuristics and thus can not guarantee the performance (i.e., how the important properties are preserved). There is a previous generation-based algorithm \cite{kim2012constructing} which can preserve in/out-degree distributions but its complexity is $O(N*E)$ ($N$: number of nodes, $E$: number of edges), making it impractical to large directed networks. 

\textbf{Present Work.} We propose a new framework of sampling large directed graphs: Constructing a sample graph using linearly rescaled JDM and DCM (Section 3.2). We show that the numbers of non-zero entries in JDM and DCM calculated from large directed network are always way smaller than the number of edges or nodes (Section 4). This property gives algorithms that adopt this framework great potential to be \textbf{efficient} because the constructing process only involves in iterating each non-zero entry of JDM and DCM. Additionally, previous work \cite{tillman2017construction} shows that graphs with the same JDM will have the \textbf{same in-degree and out-degree distributions} as the original graph and may share a lot of \textbf{similar properties} of the original graph, for example, Dyad Census, Triad Census, Shortest Path Distribution, Eigenvalues, Average Neighbor Degrees, etc. Adopting this framework, we propose a sampling algorithm that builds a simple directed graph with given JDM using the D2K construction method showed in \cite{tillman2017construction} (Section 3.2). We prove that this algorithm can preserve degree distributions (Section 5) and the deviation of distributions from the original graph has an upper bound which is negatively correlated to the sparsity of the joint degree matrix (Section 6).

\section{Related Work}

There are two lines of the past work that are related to our work.

\subsection{Construction of Directed 2K Graphs (D2K)}

The taxonomy of graph construction tasks: dK-series\cite{mahadevan2006systematic,orsini2015quantifying,erdHos2015graph} provides an elegant way to trade off accuracy (in terms of graph properties) for complexity (in generating graph realizations). The constructions of both directed and undirected dK-graphs are well understood (i.e., efficient algorithms and realizability conditions are known) for 0K (graphs with prescribed number of vertices and edges), 1K (graphs with a given degree distribution) and 2K (graphs with a given joint degree matrix). In this taxonomy system, $d$ refers to the dimension of the distribution. Specifically, the D2K algorithm we use in our sampling algorithm constructs a graph $\mathcal{G}$ with a target JDM $\mathbf{A}^{\odot}$ and a target DCM $\mathbf{B}^{\odot}$.

\subsection{Construction of Matrices with Line Sums}

Matrices with prescribed row and column sum vectors have always been an object of interest for mathematicians in the field of matrix theory. A lot of work has been done including the construction of these matrices and the necessary and sufficient conditions for such construction~\cite{brualdi2006algorithms,da2009constructing}. We use the graphical conditions and construction algorithm from~\cite{da2009constructing} as our building block of the sampling algorithm. The matrix construction algorithm takes row and column sum vectors and the integer upper bound of each entry value, $p$ as inputs and gives a matrix that satisfies these constraints as the output.

\section{Proposed Method}
\subsection{Matrices and Vectors}
We first introduce the definition of Joint Degree Matrix (JDM) and Degree Correlation Matrix (DCM).
Joint Degree Matrix, $A$ is defined as: a matrix where each entry is the number of \textbf{edges} that have the respective out-degree and in-degree pattern. Degree Correlation Matrix $B$ is defined as: a matrix where each entry is the number of \textbf{nodes} that have the respective out-degree and in-degree:
\begin{eqnarray}
{JDM} & = & \{ a_{ij} | a_{ij} = |\{(v_1, v_2) \in E | d_{v_1}^{out} = i, d_{v_2}^{in} = j \}| \} \\
{DCM} & = & \{ b_{ij} | b_{ij} = |\{v \in V | d_v^{out} = i, d_v^{in} = j\}| \}
\end{eqnarray}

\begin{displaymath}
    {JDM}  =  \{ a_{ij} | a_{ij} = |\{(v_1, v_2) \in E | deg_{out}({v_1}) = i, deg_{in}({v_2}) = j \}| \}
\end{displaymath}

\begin{displaymath}
{DCM}  =  \{ a_{ij} | a_{ij} = |\{v \in V | deg_{out}({v}) = i, deg_{in}({v}) = j\}| \}
\end{displaymath}

$A$ and $B$ can be obtained by first looping through all the edges and calculating the in/out-degrees of each node. Then, we can get $A$ and $B$ by counting the nodes and edges that have the patterns described above with another iteration.

Suppose $A$ is a matrix of size $m$ by $n$, $r_i$ and $s_j$ are the sum of row $i$ and column $j$ of $A$: $r_i = \sum_{j=1}^n a_{ij}$ and $s_j = \sum_{i=1}^m a_{ij}$, where $i = 1, \cdots m$, $j = 1, \cdots, n$.
Then $\sigma_{\mathcal{R}}(\cdot)$ and $\sigma_{\mathcal{C}}(\cdot)$ are the row and column sum functions:
\begin{eqnarray}
\sigma_{\mathcal{R}}({A}) & = & (r_1, \cdots, r_m), \\
\sigma_{\mathcal{C}}({A}) & = & (s_1, \cdots, s_n).
\end{eqnarray}

\subsection{Sampling Framework}
We propose a new graph sampling framework: Construction based Sampling with Joint Degree Matrix (JDM) and Degree Correlation Matrix (DCM). In general, this framework utilizes the favorable property: preserving JDM and DCM can not only guarantee that in/out-degree distributions will be preserved but also help capture the degree pairing patterns in edges and nodes and hence capturing more fundamental graph properties. The process of the framework is illustrated in \textbf{Figure 1} below. For a given graph, we first calculate JDM and DCM and then we conduct the sampling process: we multiply each entry of the matrices by a sample coefficient $k$ and intergerize it with functions such as $floor()$, $ceiling()$ or $round()$. After that, we use a certain construction method that can build a graph that approximately satisfies the linearly rescaled JDM and DCM. Note that the choice of construction method affects the overall performance of the algorithm a lot and in different scenarios, different construction methods may be more preferable because of the existence of accuracy-efficiency tradeoff.
\begin{figure}[h]
  \centering
  \includegraphics[width=\linewidth]{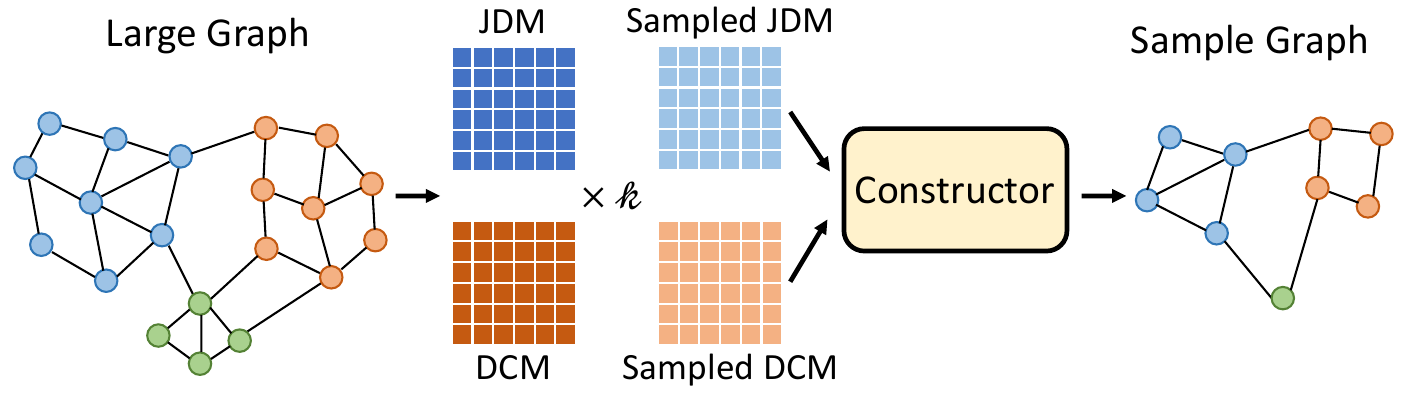}
  \caption{The process of the Construction based Sampling with JDM and DCM Framework}
\end{figure}

\subsection{Sampling Algorithm}

In this section we introduce one sampling algorithm that adopts the framework above. This algorithm uses $D2K$ as its construction method, therefore it also includes the process of adjusting JDM to satisfy the condition of $D2K$. The sampling algorithm takes Joint Degree Matrix $\mathrm{A}$, Degree Correlation Matrix $\mathrm{B}$ and Sample Coefficient $k$ as inputs and constructs a simple direct graph preserving in/out-degree distributions, joint degree distributions and degree correlation distribution (latter two are defined in Section 5.1). It proceeds by first rescales both $\mathrm{A}$ and $\mathrm{B}$ using $k$.
Then by arithmetic operations it gets the row and column sum vectors $\mathbf{r}_{\delta}$ and $\mathbf{c}_{\delta}$ as well as entry upper bound $p$ of the Adjustment Matrix $\mathrm{D}$. After that, we use $GRAPHICAL$ condition to decide if $\mathbf{r}_{\delta}$, $\mathbf{c}_{\delta}$ and $p$ are realizable. If $GRAPHICAL$ returns $\mathbf{TRUE}$ we use $CONSTRUCT$ to build such $\mathrm{D}$ and add it to linearly-rescaled Joint Degree Matrix $\mathrm{A}$. Finally, we use $D2K$ algorithm to build a simple, directed sample graph $\mathcal{G'}$ with target Joint Degree Matrix (after modification) $\mathrm{A}^{\odot}$ and target Degree Correlation Distribution Matrix $\mathrm{B}^{\odot}$. Note that detailed description of $GRAPHICAL$, $CONSTRUCT$ and $D2K$ can be found in \cite{da2009constructing} and \cite{tillman2017construction}.
\begin{algorithm}
\SetKwInOut{Input}{input}
\SetKwInOut{Output}{output}

\Input{Joint Degree Matrix $\mathrm{A}$, Degree Correlation Matrix $\mathrm{B}$, Sample Coefficient $k$}
\Output{A simple directed sample graph $\mathcal{G'}$ preserving four distributions}
\caption{Sampling Algorithm}


$A' \gets \{a_{ij}'|a_{ij}' = \lfloor \frac{1}{k} a_{ij} \rfloor, \forall a_{ij} \in \mathrm{A} \} $

$B' \gets \{b_{ij}' | b_{ij}' = \lceil \frac{1}{k} b_{ij} \rceil, \forall b_{ij} \in \mathrm{B}\} $

$\mathbf{r}_{A'} \gets \sigma_{\mathcal{R}}(\mathrm{A'})$, 
$\mathbf{r}_{B'} \gets \sigma_{\mathcal{R}}(\mathrm{B'})$, 
$\mathbf{c}_{A'} \gets \sigma_{\mathcal{C}}(\mathrm{A'})$, 
$\mathbf{c}_{B'} \gets \sigma_{\mathcal{C}}(\mathrm{B'})$

\For{$i = 1, \cdots, m$}{
$\widetilde{\mathbf{r}_{B'}}^{(i)} := \mathbf{r}_{B'}^{(i)} \cdot i$
}

\For{$j = 1, \cdots, n$}{
$\widetilde{\mathbf{c}_{B'}}^{(j)} := \mathbf{c}_{B'}^{(j)} \cdot j$
}

$\mathbf{r}_{\delta} \gets \widetilde{\mathbf{r}_{B'}} - \mathbf{r}_{A'}$, $\mathbf{c}_{\delta} \gets \widetilde{\mathbf{c}_{B'}} - \mathbf{c}_{A'}$

$\mathrm{L} \gets \{l_{ij} | l_{ij} = \mathrm{r}^{(i)}_{B'} \cdot \mathrm{c}^{(j)}_{B'} - b'_{ij}, \forall b'_{ij} \in B' \}$

$p := \mathbf{min}_{l_{ij}\in L}\{ l_{ij} - a'_{ij}\}$

\eIf{$GRAPHICAL(\mathbf{r}_{\delta}, \mathbf{c}_{\delta}, p) == \mathbf{TRUE}$}{
$D \gets CONSTRUCT(\mathbf{r}_{\delta}, \mathbf{c}_{\delta}, p)$
}
{
\textbf{return} $\mathbf{FALSE}$
}

$\mathrm{A}^{\odot} \gets \mathrm{A}' + \mathrm{D}$,
$\mathrm{B}^{\odot} \gets \mathrm{B}'$

$\mathcal{G'(V',E')} \gets D2K(\mathrm{A}^{\odot}, \mathrm{B}^{\odot})$

\textbf{return} $\mathcal{G'}$ 
\label{alg: algo1}
\end{algorithm}

\section{Sparsity of JDM and DCM}

\begin{table*}[]
\begin{tabular}{|c|c|c|c|c|c|c|c|}
\hline
Category                                & Name            & N       & E       & \# DCM & \# JDM & \% DCM  & \% JDM  \\ \hline
\multirow{5}{*}{Online Social Networks} & FilmTrust trust & 874     & 1853    & 102    & 413    & 11.67\% & 22.29\% \\ \cline{2-8} 
                                        & CiaoDVD trust   & 4658    & 40133   & 834    & 11127  & 17.90\% & 27.73\% \\ \cline{2-8} 
                                        & Epinions        & 75879   & 508837  & 4398   & 96382  & 5.80\%  & 18.94\% \\ \cline{2-8} 
                                        & Twitter (ICWSM) & 465017  & 834797  & 1172   & 25501  & 0.25\%  & 3.05\%  \\ \cline{2-8} 
                                        & Youtube links   & 1138499 & 4942297 & 8859   & 322950 & 0.78\%  & 6.53\%  \\ \hline
\multirow{3}{*}{Citation Networks}      & DBLP            & 12590   & 49759   & 769    & 5201   & 6.11\%  & 10.45\% \\ \cline{2-8} 
                                        & arXiv hep-ph    & 34546   & 421578  & 3925   & 28106  & 11.36\% & 6.67\%  \\ \cline{2-8} 
                                        & CiteSeer        & 384413  & 1751463 & 4030   & 27399  & 1.05\%  & 1.56\%  \\ \hline
\end{tabular}
\caption{The Number of non-negative entries in JDM and DCM compared to the number of edges and nodes in real-world networks}
\label{table1}
\end{table*}

In the proposed sampling algorithm, we use two matrices frequently: Joint Degree Matrix $A$ and Degree Correlation Matrix $B$. As mentioned above, these two matrices are always quite sparse. We illustrate this property by conducting experiments on a variety of different real-world networks that belong to different categories and have different sizes. All datasets are from a free public project called \textbf{The KONECT Project} \cite{kunegis2013konect}. The experimental data are demonstrated in \textbf{Table 1}.

Experimental datasets include five online social networks and three citation networks with node sizes ranging from $10^3$ to $10^7$ and edge sizes ranging from $10^4$ to $10^7$. The meaning of each column is, $N$: the number of nodes, $E$: the number of edges, \# DCM: the number non-zero entris in Degree Correlation Matrix, \# JDM: the number of non-zero entris in Joint Degree Matrix, \% DCM: $\frac{\# DCM}{N}$, \% JDM: $\frac{\# JDM}{E}$.

From the data we can see that \% DCM and \% JDM are below $30\%$ for all datasets, and as the size of graph grows bigger, these percentages tend to drop. For datasets that have a node size larger than $10^6$, the percentages are around 1\%. This observation shows that JDM and DCM are usually quite sparse and the number of non-zero entries in them are quite small compared to the size of nodes and edges, especially when the graph size exceeds certain threshold.

\section{Proof of Validity}

\subsection{Notations}
First, we introduce some notations that will be used in the following sections. $\mathcal{V} = \{v_i\}$ is the set of vertices. $\mathcal{E} \subseteq \{(v_1,v_2) |(v_1, v_2) \in \mathcal{V}^2, v_1 \neq v_2 \}$ denotes the directed edge set where each element belongs to the Cartesian square of set $\mathcal{V}$. $(v_1, v_2)$ represents an edge that is pointing from $v_1$ to $v_2$. $|\mathcal{V}|$ and $|\mathcal{E}|$ are number of vertices and edges in the graph. 
We define $V_{k,p} = \{v \in \mathcal{V} | d_v^p = k\} \subset \mathcal{V}$, $p \in \{in, out \}$ as a vertices subset of $\mathcal{V}$, where all vertices' in-degree (or out-degree, depends on the value of $p$) equals to k. For example, $V_{1, in}$ denotes the subset of nodes with in-degree 1.
Note that in order to avoid confusion we doesn't use $N$ and $E$ to represent the vertices and edges number in section 5 and 6.

We denote $P(k,p) = \frac{|V_{k,p}|}{|\mathcal{V}|}$ where $k = 0,1,2,\cdots$; $p \in \{in, out \}$ as the degree distribution, which equals to the fraction of nodes with certain in-degree (or out-degree). So, $P(k, in)$ is the \textbf{in-degree distribution} and $P(k, out)$ is the \textbf{out-degree distribution}. For example, $P(1, in) = \frac{|V_{1, in}|}{|\mathcal{V}|}$ denotes for the fraction of nodes with in-degree 1.

Furthermore, we define $P(i,j) = \frac{|\{ v \in \mathcal{V} | d_v^{out} = i, d_v^{in} = j\}|}{|\mathcal{V}|}$, $i,j = 0,1,2,\cdots$ as the \textbf{degree correlation distribution}, which equals to the fraction of nodes having certain out-degree and in-degree. Given a pair of nodes $v_i$ and $v_j$ that are connected, we denotes $\tilde{P}(i, j) = \frac{|\{ (v_1, v_2) \in \mathcal{E} | d_{v_1}^{out} = i, d_{v_2}^{in} = j\}|}{|\mathcal{E}|}$ as their \textbf{joint degree distribution}.

\subsection{Consistency}
In this section we will prove the following important fact:
\begin{itemize}
    \item If the $GRAPHICAL$ checking gives us $\mathbf{TRUE}$, the condition of $D2K$ will be automatically satisfied.
\end{itemize}

In \cite{tillman2017construction}, Balint Tillman et al. give the condition for target JDM: $\mathrm{A}^{\odot}$ and target DCM: $ \mathrm{B}^{\odot}$ to be realizable(i.e. graphical). We give the equivalent D2K CONDITION, which is a modified version from their original condition to match the changes of definition and notation in this paper. 
\begin{theorem}[D2K CONDITION]
Let $\mathrm{A}^{\odot}$ be the joint degree matrix, $\mathrm{B}^{\odot}$ be the degree correlation matrix both of size $m$ by $n$. There is a graph $\mathcal{G}$ satisfying $\mathrm{B} = \mathrm{B}^{\odot}$ and $\mathrm{A} = \mathrm{A}^{\odot}$ if and only if
$\forall i = 1, \cdots, m ; j= 1, \cdots, n$,

\begin{equation}
|V_{i, out}| = \sum_j \frac{a^{\odot}_{ij}}{i} = \sum_j b^{\odot}_{ij},
\end{equation}
\begin{equation}
|V_{j, in}| = \sum_i \frac{a^{\odot}_{ij}}{j} = \sum_i b^{\odot}_{ij}.
\end{equation}
\begin{equation}
a^{\odot}_{ij} + b^{\odot}_{ij} \leq |V_{i, out}| \cdot |V_{j, in}|,
\end{equation}
\end{theorem}

We show that the D2K CONDITION will be automatically satisfied by the process of our sampling algorithm if $GRAPHICAL$ gives $\mathbf{TRUE}$. 

Next, we divide the task into two lemmas and provide proofs for each of them.

\begin{lemma}[Condition 1 and 2]
If $GRAPHICAL(\mathbf{r}_{\delta}, \mathbf{c}_{\delta}, p) == \mathbf{TRUE}$, by following the steps of \textbf{Algorithm 1}, the first and second D2K CONDITION will be satisfied.
\end{lemma}

\begin{proof}
Equation (5) and (6) are equivalent to the following identity: 
\begin{equation}
    \sum_j a^{\odot}_{ij} = i \cdot \sum_j b^{\odot}_{ij}
\end{equation}
\begin{equation}
    \sum_i a^{\odot}_{ij} = j \cdot \sum_i b^{\odot}_{ij}
\end{equation}

Without loss of generality, we only need to show (8)
holds for an arbitrary choice of $i$ and $j$.
Note that we have:
\begin{displaymath}
a^{\odot}_{ij} = a_{ij}' + d_{ij}, b^{\odot}_{ij} = b_{ij}'
\end{displaymath}
\begin{equation}
\therefore (8) \longleftrightarrow \sum_j a_{ij}' + \sum_j d_{ij} = i \cdot \sum_j b_{ij}'
\end{equation}
According to \textbf{Algorithm 1}, (10) is equivalent to:
\begin{equation}
\mathbf{r}_{A'}^{(i)} + \mathbf{r}_{\delta}^{(i)} = i \cdot \mathbf{r}_{B'}^{(i)}
\end{equation}
\begin{displaymath}
\because \mathbf{r}_{\delta}^{(i)} = \widetilde{\mathbf{r}_{B'}^{(i)}} - \mathbf{r}_{A'}^{(i)} = i \cdot \mathbf{r}_{B'}^{(i)} - \mathbf{r}_{A'}^{(i)}
\end{displaymath}
\begin{equation}
\therefore \mathbf{r}_{A'}^{(i)} + \mathbf{r}_{\delta}^{(i)} = \mathbf{r}_{A'}^{(i)} + i \cdot \mathbf{r}_{B'}^{(i)} - \mathbf{r}_{A'}^{(i)} = i \cdot \mathbf{r}_{B'}^{(i)}
\end{equation}
This completes the proof of LEMMA 5.2

\end{proof}

\begin{lemma}[Condition 3]
If $GRAPHICAL(\mathbf{r}_{\delta}, \mathbf{c}_{\delta}, p) == \mathbf{TRUE}$, following the steps of \textbf{Algorithm 1}, the third D2K CONDITION will be satisfied.
\end{lemma}

\begin{proof}
From the proof of LEMMA 5.2, we know that:
\begin{displaymath}
\because |V_{i, out}| \cdot |V_{j, in}| = ij \cdot \sum_j b_{ij}^{\odot} \cdot \sum_i b_{ij}^{\odot} = ij \cdot \sum_j b_{ij}' \cdot \sum_i b_{ij}'
\end{displaymath}
and 
\begin{displaymath}
a^{\odot}_{ij} + b^{\odot}_{ij} = a_{ij}' + d_{ij} + b_{ij}'
\end{displaymath}
Thus, (7) is equivalent to:
\begin{equation}
d_{ij} \leq ij \cdot \sum_j b_{ij}' \cdot \sum_i b_{ij}' - b_{ij}' - a_{ij}' = l_{ij} - a_{ij}'
\end{equation}
According to algorithm $CONSTRUCT$, we have
\begin{displaymath}
d_{ij} \leq p = \mathbf{min}_{l_{ij}\in L}\{ l_{ij} - a'_{ij}\} 
\end{displaymath}
\begin{displaymath}
\therefore d_{ij} \leq l_{ij} - a'_{ij}, \forall l_{ij}\in L
\end{displaymath}
This completes the proof for LEMMA 5.3

\end{proof}

\subsection{Preserving Distributions}

In this section we show that the sample graph has the \textbf{same in/out-degree distribution} and the \textbf{same degree correlation distribution} as the original graph $\mathcal{G}$, i.e. $P(k,in)$, $P(k,out)$ and $P(i,j)$. Additionally, we will also show that the joint degree distribution $\tilde{P}'(i, j)$ of sample graph will also be similar to the joint degree distribution $\tilde{P}(i, j)$ of the original graph with an upper bound that will be studied in Section 6.

We construct two important variable $\mathring{a_{ij}}$ and $\mathring{b_{ij}}$ as below:
\begin{equation}
    \mathring{a_{ij}} = \frac{1}{k} a_{ij}, \mathring{b_{ij}} = \frac{1}{k} b_{ij}.
\end{equation}
By definition of $a_{ij}'$ and $b_{ij}'$ we have:
\begin{equation}
a_{ij}' = \lfloor \frac{1}{k} a_{ij} \rfloor = \lfloor \mathring{a_{ij}} \rfloor, b_{ij}' = \lceil \frac{1}{k} b_{ij} \rceil = \lceil \mathring{b_{ij}} \rceil
\end{equation}
We get two important inequalities:
\begin{equation}
    \mathring{a_{ij}} -1 < a_{ij}' \leq \mathring{a_{ij}}, \mathring{b_{ij}} \leq b_{ij}' < \mathring{b_{ij}} + 1
\end{equation}
Firstly, we show that the degree correlation distribution is preserved.
From the definition of $P(i,j)$ and Degree Correlation Matrix (DCM) $B$ we have:
\begin{equation}
P(i,j) = \frac{|\{ v \in V | d_v^{out} = i, d_v^{in} = j\}|}{|V|} = \frac{b_{ij}}{\sum_{i,j}b_{ij}} = \frac{\frac{1}{k}b_{ij}}{\frac{1}{k}\sum_{i,j}b_{ij}} 
\end{equation}
From step 18 in \textbf{Algorithm 1} we have:
\begin{equation}
a^{\odot}_{ij} = a_{ij}' + d_{ij}, b^{\odot}_{ij} = b_{ij}'
\end{equation}
By 
(8), (9), (14) and (15):
\begin{displaymath}
P(i,j) =
\frac{\frac{1}{k}b_{ij}}{\frac{1}{k}\sum_{i,j}b_{ij}} = \frac{\mathring{b_{ij}}}{\sum_{i,j}\mathring{b_{ij}}} \approx \frac{b_{ij}'}{\sum_{i,j}b_{ij}'} = \frac{b_{ij}^{\odot}}{\sum_{i,j}b_{ij}^{\odot}}
= P^{\odot}(i,j)
\end{displaymath}

Thus the sample graph $\mathcal{G}'$ has the same degree correlation distribution and the only deviation is from the Integerization process. This deviation is quantified in Section 6.

Note that $P(i,j)$ is the joint probability mass function of $P(k,in)$ and $P(k,out)$:
\begin{displaymath}
P(k,in) = P(\cdot, k) = \sum_i P(i,k); P(k,out) = P(k, \cdot) = \sum_j P(k,j).
\end{displaymath}

Therefore the sample graph $\mathcal{G}'$ will automatically also preserve the in/out-degree distribution of the original graph when the degree correlation distribution is preserved.

From the definitions of joint degree distribution and Joint Degree Matrix (JDM), we have the following observation:
\begin{equation}
\tilde{P}(i, j) = \frac{|\{ (v_1, v_2) \in E | d_{v_1}^{out} = i, d_{v_2}^{in} = j\}|}{|E|} = \frac{a_{ij}}{\sum_{i,j}a_{ij}}.
\end{equation}
Because of (8), (14) and (15)
, similarly, we have 
\begin{displaymath}
\tilde{P}(i, j)
=
\frac{a_{ij}}{\sum_{i,j}a_{ij}} = \frac{\frac{1}{k}a_{ij}}{\sum_{i,j}\frac{1}{k}a_{ij}} = \frac{\mathring{a_{ij}}}{\sum_{i,j}\mathring{a_{ij}}}
\approx \frac{a_{ij}'}{\sum_{i,j}a_{ij}'}.
\end{displaymath}

Because 
$a^{\odot}_{ij} = a_{ij}' + d_{ij}$, the joint degree distribution of sample graph $\tilde{P}^{\odot}(i, j) = \frac{a_{ij}^{\odot}}{\sum_{i,j}a_{ij}^{\odot}}$ is similar but different from the joint degree distribution of the original graph$\frac{a_{ij}'}{\sum_{i,j}a_{ij}'}$. In next section, we will show that the deviation caused by $d_{ij}$ also has an upper bound.

\section{Deviation Analysis}
In this section, we attempt to quantify the deviations of degree distributions by giving an upper bound for them. 

\subsection{Integerization Deviation}

From (16), we can derive the following inequalities:
\begin{equation}[\textbf{Deviation}]
    \frac{\mathring{b_{ij}}}{\sum_{i,j}\mathring{b_{ij}} + mn -1} 
    < 
    \frac{b_{ij}'}{\sum_{i,j}b_{ij}'} 
    < 
    \frac{\mathring{b_{ij}} + 1}{\sum_{i,j}\mathring{b_{ij}} + 1}
\end{equation}
and
\begin{equation}[\textbf{Deviation}]
    \frac{\mathring{a_{ij}} - 1}{\sum_{i,j}\mathring{a_{ij}} - 1}
    < 
    \frac{a_{ij}'}{\sum_{i,j}a_{ij}'} 
    < 
    \frac{\mathring{a_{ij}}}{\sum_{i,j}\mathring{a_{ij}} - mn +1} 
\end{equation}

\begin{displaymath}
    \left(\frac{\mathring{b_{ij}}}{\sum_{i,j}\mathring{b_{ij}} + mn -1} , \frac{\mathring{b_{ij}} + 1}{\sum_{i,j}\mathring{b_{ij}} + 1}\right)
\end{displaymath}

\begin{displaymath}
    \left(\frac{\mathring{a_{ij}} - 1}{\sum_{i,j}\mathring{a_{ij}} +mnp -p - 1}, \frac{\mathring{a_{ij}}+p}{\sum_{i,j}\mathring{a_{ij}} +p - mn +1} \right)
\end{displaymath}

where $m$ is the number of rows and $n$ is the number of columns.
From \textbf{Algorithm 1} we know that $b_{ij}^{\odot} = b_{ij}'$. Thus (20) quantifies the deviation brought by Integerization for degree correlation distribution $P^{\odot}(i,j)$. 
\begin{equation}
    \frac{\mathring{b_{ij}}}{\sum_{i,j}\mathring{b_{ij}} + mn -1} 
    < 
    \frac{b_{ij}^{\odot}}{\sum_{i,j}b_{ij}^{\odot}} 
    =
    P^{\odot}(i,j)
    < 
    \frac{\mathring{b_{ij}} + 1}{\sum_{i,j}\mathring{b_{ij}} + 1}
\end{equation}
From \textbf{Algorithm 1} we know that $a^{\odot}_{ij} = a_{ij}' + d_{ij}$. Therefore, for joint degree distribution $\tilde{P}^{\odot}(i, j)$, however, we still need to consider the change in value brought by $d_{ij}$. We know from algorithm CONSTRUCT that $0 \leq d_{ij} \leq p$
\begin{displaymath}
    \frac{a_{ij}'}{\sum_{i,j}a_{ij}' + mnp -p}
    <
    \frac{a_{ij}'+d_{ij}}{\sum_{i,j}a_{ij}'+\sum_{i,j}d_{i,j}}
    =
    \frac{a_{ij}^{\odot}}{\sum_{i,j}a_{ij}^{\odot}} 
    <
    \frac{a_{ij}' + p}{\sum_{i,j}a_{ij}' + p}
\end{displaymath}
Combining the above inequalities with (21), we have:
\begin{equation}
    \frac{\mathring{a_{ij}} - 1}{\sum_{i,j}\mathring{a_{ij}} +mnp -p - 1}
    < 
    \frac{a_{ij}^{\odot}}{\sum_{i,j}a_{ij}^{\odot}} 
    =
    \tilde{P}^{\odot}(i,j)
    < 
    \frac{\mathring{a_{ij}}+p}{\sum_{i,j}\mathring{a_{ij}} +p - mn +1} 
\end{equation}
This quantifies the deviation of $\tilde{P}^{\odot}(i, j)$ from $\tilde{P}(i, j)$.

\subsection{Effects of Sparsity}

Moreover, if we take the sparsity of the graph into account, we can get a more accurate bound of deviation. This is because the deviation of Intergerization does not happen to those entries that are zero in original matrices ($a_{ij}' = \mathring{a_{ij}}, \forall a_{ij} = 0$, $b_{ij}' = \mathring{b_{ij}}, \forall b_{ij} = 0$). Hence, the sparsity of the original JDM and DCM will directly affect the deviation from Intergerization.

We define the {sparsity coefficient} of row $i$ in matrix $A_{m \times n}$ as:
\begin{displaymath}
    s_{\mathcal{R}}^A(i) = \frac{I_{(j \in \{1, \cdots, n\})} a_{ij} = 0}{n}.
\end{displaymath}
Similarly the sparsity coefficient of column $j$ in matrix $A_{m \times n}$ is:
\begin{displaymath}
    s_{\mathcal{C}}^A(j) = \frac{I_{(i \in \{1, \cdots, m\})} a_{ij} = 0}{m}.
\end{displaymath}
$I$ is the indicator function. Note that greater the {sparsity coefficient} is, more sparse that line/column will be (i.e. the fraction of 0 in that line/column).
Using the definition of {sparsity coefficient}, we can refine the qualification inequalities (20), (23) as:
\begin{equation}
    \frac{\mathring{b_{ij}}}{\sum_{i,j}\mathring{b_{ij}} + m_B'\cdot n_B'-1} 
    < 
    P^{\odot}(i,j)
    < 
    \frac{\mathring{b_{ij}} + 1}{\sum_{i,j}\mathring{b_{ij}} + 1},
\end{equation}
and
\begin{equation}
    \frac{\mathring{a_{ij}} - 1}{\sum_{i,j}\mathring{a_{ij}} +m_A' n_A'p -p - 1}
    < 
    \tilde{P}^{\odot}(i,j)
    < 
    \frac{\mathring{a_{ij}}+p}{\sum_{i,j}\mathring{a_{ij}} +p - m_A' n_A' +1},
\end{equation}
where

$m_A' = m \cdot (1-s_{\mathcal{C}}^A(j))$, $n_A' = n \cdot (1-s_{\mathcal{R}}^A(i))$

$m_B' = m \cdot (1-s_{\mathcal{C}}^B(j))$, $n_B' = m \cdot (1-s_{\mathcal{R}}^B(i))$

Note that greater the {sparsity coefficient}, smaller the deviation of both $P^{\odot}(i,j)$ from $P(i,j)$ and $\tilde{P}^{\odot}(i, j)$ from $\tilde{P}(i, j)$. This property has strong realistic meaning because we show in section 4 that JDM and DCM are always quite sparse and they tend to get more sparse as the size of the graph increases. Therefore the real performance of the sampling algorithm proposed in this paper will be better than the range presented in Section 6.1.




\section{Conclusions}

We propose a new sampling framework that is efficient and is able to preserve important graph properties. Based on this framework we provide a new sampling algorithm using D2K construction method. We prove that this algorithm can preserve in/out-degree distributions, joint degree distributions and degree correlation distributions. We also analyze the effects of the JDM, DCM sparsity on deviation of degree distributions and provide upper bounds that are modified with sparsity coefficients for deviations. Additionally, we use experiments to show that JDM and DCM of real-world graphs are always sparse, which lends credence to the belief that the proposed sampling algorithm will have a better-than-theory performance on real-life large directed networks. Finally, it is worth pointing out that by utilizing more efficient construction algorithms the potential of the framework may be more thoroughly realized. Hence, future work on finding faster and more accurate construction algorithms with JDM and DCM is worth conducting.

\begin{acks}
This work was supported in part by NSF IIS-1849816, IIS-2142827, IIS-2146761, and ONR N00014-22-1-2507.
\end{acks}

\section*{Code Availability}
The source code for this paper is openly available at \url{https://github.com/RaccoonOnion/sample}.

\bibliographystyle{ACM-Reference-Format}
\bibliography{main}

\end{document}